\documentclass[aps,pra,twocolumn,amsfonts,amssymb,amsmath,showpacs,
floatfix,nofootinbib,citesort]{revtex4-2}
\usepackage{mathrsfs}
\usepackage{amsfonts}
\usepackage{amstext}
\usepackage{amsmath}
\usepackage{amssymb}
\usepackage{bm}
\usepackage{CJK}
\usepackage{mathtools}
\usepackage{amssymb}
\usepackage{amsthm}\usepackage{bbm}
\usepackage[dvips]{graphicx}
\def\qed{\leavevmode\unskip\penalty9999 \hbox{}\nobreak\hfill
	\quad\hbox{\leavevmode  \hbox to.77778em{%
			\hfil\vrule   \vbox to.675em%
			{\hrule width.6em\vfil\hrule}\vrule\hfil}}
	\par\vskip3pt}

\usepackage{amssymb}
\usepackage{graphicx}
\usepackage{graphics}
\usepackage{amsmath}
\usepackage{amsthm}
\usepackage{color}
\usepackage{dsfont}
\usepackage{textcomp}
\definecolor{darkred}  {rgb}{0.5,0,0}
\definecolor{darkblue} {rgb}{0,0,0.5}
\definecolor{darkgreen}{rgb}{0,0.5,0}
\usepackage{hyperref}
\hypersetup{
	pdftitle = {QRT Proposal},
	pdfauthor = {},
	colorlinks = true,
	urlcolor  = blue,         
	linkcolor = red,     
	citecolor = blue,    
	filecolor = darkred       
}
\usepackage{mathtools}
\def\ra{\rangle}
\def\la{\langle}

\def\ot{\otimes}
\newtheorem{theorem}{Theorem}
\newtheorem{conjecture}{Conjecture}

\newtheorem{pro}{Proposition}

\newcommand{\bea}{\begin{eqnarray}}
	\newcommand{\eea}{\end{eqnarray}}
\newcommand{\be}{\begin{equation}}
	\newcommand{\ee}{\end{equation}}
\newcommand{\ba}{\begin{equation}\begin{aligned}}
		\newcommand{\ea}{\end{aligned}\end{equation}}
	\newcommand{\bax}{\begin{equation*}\begin{aligned}}
			\newcommand{\eax}{\end{aligned}\end{equation*}}
\newcommand{\rank}{\text{rank}}

\newcommand{\beax}{\begin{eqnarray*}}
	\newcommand{\eeax}{\end{eqnarray*}}
\newcommand{\bex}{\begin{equation*}}
	\newcommand{\eex}{\end{equation*}}

\newtheorem{definition}{Definition}
\theoremstyle{remark}
\newtheorem{remark}{Remark}

\def\be{\begin{equation}}
	\def\ee{\end{equation}}

\newcommand{\mE}{\mathcal{E}}
\newcommand{\mC}{\mathcal{C}}

\newcommand{\mI}{\mathcal{I}}
\newcommand{\vmI}{\check{\mathcal{I}}}
\newcommand{\mH}{\mathcal{H}}

\newcommand{\mP}{\mathcal{P}}

\newcommand{\mR}{\mathcal{R}}

\newcommand{\mS}{\mathcal{S}}

\newcommand{\tr}{{\rm Tr}}

\newcommand{\mbb}[1]{\mathbb{#1}}

\newcommand{\ket}[1]{|#1\rangle}

\newcommand{\mbR}{\mathbb{R}}

\newcommand{\rIm}{{\rm Im}}
\newcommand{\rmi}{{\rm i}}
\newcommand{\bfr}{{\bf r}}
\newcommand{\bfc}{{\bf c}}
\newcommand{\bfp}{{\bf p}}
\newcommand{\bfq}{{\bf q}}
\newcommand{\bfu}{{\bf u}}
\newcommand{\bfv}{{\bf v}}
\newcommand{\bfs}{{\bf s}}
\newcommand{\bfx}{{\bf x}}

\begin{document}

\preprint{APS/123-QED}
\begin{CJK*}{GB}{gbsn}
\title{Measure of set imaginarity\\}
\author{Yu Guo$^{1,2}$}
\email{guoyu3@aliyun.com}
\author{Jiabo Pan$^{1,2}$}
\author{Yuqin Wang$^{1,2}$}
\author{Shuanping Du$^{3}$}

\affiliation{$^{1}$School of Mathematical Sciences, Inner Mongolia University, Hohhot, Inner Mongolia 010021, People's Republic of China}
\affiliation{$^{2}$Inner Mongolia Key Laboratory of Mathematical Modeling and Scientific Computing, Inner Mongolia University, Hohhot, Inner Mongolia 010021, People's Republic of China}
\affiliation{$^{3}$School of Mathematical Sciences, Xiamen University, Xiamen, Fujian, 361000, People's Republic of China}


\begin{abstract}
Recent studies have shown that Bargmann invariants provide effective detectors of set imaginarity. In this paper, we investigate set imaginarity as a quantum resource in qubit systems. By exploiting the structure of Bargmann invariants, we show that the free operations for qubit set imaginarity consist precisely of common unital operations and common planarized operations. Based on this characterization, we introduce an axiomatic framework for set-imaginarity measures (SIMs). In particular, we propose two refined notions, namely unified SIMs and complete SIMs, which allow a more fine-grained quantification of set imaginarity. To make these notions concrete, we construct two qubit SIMs from the Bargmann invariants of three-state subsets. We prove that one of them is a unified SIM, while the other satisfies the stronger requirements of a complete SIM. Furthermore, we revisit the robustness of set imaginarity previously introduced in the literature. We show that, although this robustness is a valid SIM for qubit systems, it is neither a unified SIM nor a complete SIM. To overcome this limitation, we propose an improved robustness-type measure and rigorously prove that it defines a complete qubit SIM.

\end{abstract}

\maketitle
\end{CJK*}


\section{Introduction}


Quantum resource theory has traditionally focused on the characterization and manipulation of individual quantum states or quantum channels~\cite{Chitambar2019rmp,Gour2024arxiv,Coecke2016ic,Guo2025arxiv} and has developed into a well-established framework within quantum information science. However, this single-object paradigm becomes insufficient
for describing a wide range of intrinsically collective phenomena in quantum physics. In such settings, as exemplified by quantum pseudorandomness~\cite{Ji2018,Kretschmer2021,Bansal2025prx}, quantum non-Markovianity~\cite{Rivas2014rpp}, and linear-optical multiport interferometry~\cite{Shchesnovich2015pra,Shchesnovich2018pra,Jones2020,Jones2023pra}, a proper physical description requires considering sets or structured families of quantum objects rather than isolated instances.

Consequently, quantum resources associated with sets of states  have recently emerged as an active topic of research~\cite{Designolle2021prl,Miyazaki2022q,Sajjan,Buscemi2020prl,Uola2019prl,Martins2020pra,Ducuara2020prr,Selby2023prl,Wagner2024pra,Galvao2020pra,Zhang2025arxiv,Zhang2025arxiv2,Gour2018}. Among the various developments, two notable examples are \emph{set coherence} and \emph{set imaginarity}, which extend the corresponding well-established single-state resource theories respectively. Set coherence was formally introduced by Designolle \emph{et al.}~\cite{Designolle2021prl}, although related notions had previously appeared under different terminologies~\cite{Horodecki2007pra,Galvao2020pra,Piani2014njp,Fuchs2003qic,Horodecki2006ijqi}. In parallel,  set imaginarity was developed by Miyazaki and Matsumoto~\cite{Miyazaki2022q}. A fundamental distinction between resource theories of individual states and their set-based counterparts lies in the role played by the underlying reference structure: While single-state resource theories typically rely on a fixed reference basis, set-based resource theories are often formulated in a basis-independent manner, with the relevant question being whether a common reference basis exists for the entire set. In particular, any set consisting entirely of mutually commuting states exhibits neither set coherence nor set imaginarity since all states in the set can be simultaneously diagonalized in a common basis.

The resource theory of imaginarity has been shown to possess a precise operational interpretation~\cite{WuKD2021prl,WuKD2021pra,WuKD2023arxiv}, indicating that imaginarity can be linked to concrete advantages in quantum information tasks. Furthermore, imaginarity has been demonstrated to provide advantages in a variety of quantum information tasks. These include state discrimination~\cite{WuKD2021prl,Herzog2002pra}, hiding and masking~\cite{ZhuH2021prr}, quantum machine learning~\cite{Sajjan},
pseudorandomness~\cite{Koh2023arxiv}, multiparameter quantum metrology~\cite{Carollo2018srep,Carollo2019jsm,Miyazaki2022q},
linear-optical implementations~\cite{Jones2023pra,Menssen2017prl,Shchesnovich2018pra}, Kirkwood-Dirac quasiprobabilities~\cite{Wagner2024qst,Budiyono2023pra1,Budiyono2023pra2,Budiyono2023jmp}, and weak-value amplification and related phenomena~\cite{Wagner2023pra,Kedem2012pra,Dixon2009prl,Hosten2008science,Brunner2010prl,Hofmann2011pra,Kunjwal2019pra}. These results indicate that the imaginary components of density matrices are not merely representational artifacts but may encode genuine physical and informational content.

To characterize such basis-independent imaginary features, the \emph{Bargmann invariant} provides a particularly useful tool.
Originally introduced by Bargmann in the study of symmetry transformations~\cite{Bargmann1964}, the $n$th order Bargmann invariant associated with $\vec{\rho}=\{\rho_1,\rho_2,\ldots,\rho_n\}$ is defined as
\bea
\Delta_n(\vec{\rho})
=
\tr(\rho_1\rho_2\cdots\rho_n),
\eea
where $\rho_j$s are density operator on a $d$-dimensional Hilbert space $\mH$ (the original Bargmann invariant refers to $\Delta_3(|\psi_1\ra,|\psi_2\ra,|\psi_3\ra)=\la\psi_1|\psi_2\ra\la\psi_2|\psi_3\ra\la\psi_3|\psi_1\ra$ for pure states $\{|\psi_1\ra, |\psi_2\ra, |\psi_3\ra\}$). A key property of the Bargmann invariant is its invariance under simultaneous unitary transformations, i.e.,
\bea
\Delta_n(\vec{\rho})
=
\Delta_n(U\vec{\rho}U^\dag),\quad\forall~\text{unitary}~U,
\eea
where $U\vec{\rho}U^\dag:=\{U\rho_1U^\dag,\ldots,U\rho_nU^\dag\}$, which makes it a natural candidate for detecting basis-independent features of a set of states. In particular, $\Delta_n(\vec{\rho})\notin\mathbb{R}$ implies that the set $\vec{\rho}$ necessarily has set imaginarity. Unlike pairwise overlaps, they encode higher-order correlations that are intrinsic to the configuration of the states, which have led to, in addition to detecting set imaginarity~\cite{Fernandes2024prl,LiTan2025}, a wide range of applications~\cite{Berry1984,Mukunda1993I,Menssen2017prl,Jones2020,Liang2023,Reascos2023,Wagner2023pra,Kirkwood1933,Dirac1945,Arvidsson2024}. Very recently, Li \textit{et al.}~\cite{LiMS2026pra} showed that a set of qubit states exhibits set imaginarity if and only if there exists a triplet of states whose third-order Bargmann invariant has a nonzero imaginary part (see Corollary III.7 in Ref.~\cite{LiMS2026pra}).

The conceptual framework of quantum resource theory is built upon three fundamental pillars:
(i) a specification of a set of free states, with states outside this set regarded as resource states;
(ii) a set of free operations, i.e., the resource nongenerating operations;
and (iii)  a rigorous quantitative measure for resource assessment.
Within the specific context of set imaginarity, current theoretical understanding has elucidated the structural characteristics of resource states through Bargmann invariants analysis. 

Nevertheless, a systematic characterization of the corresponding free operations continues to represent a significant gap in contemporary research. A fundamental axiomatic requirement dictates that any well-defined resource measure must exhibit monotonic nonincrease under admissible free operations. In this article, we aim to identify the class of free operations for set imaginarity in qubit systems and then to establish an axiomatic foundation for set-imaginarity measures (SIMs). Building upon the properties encoded in Bargmann invariants for qubit states, we propose two distinct SIMs. In addition, a rigorous verification is performed to assess if the SIMs so far meet all requisite conditions in our axiomatic definition. We also discuss other SIMs which are the renovations of the corresponding individual state imaginarity measures, respectively.

Throughout this paper, we denote by $\mS(d)$ the set of all quantum states acting on the system described by the Hilbert space $\mH$ with $\dim\mH=d\geqslant2$. For a set of states $\vec{\rho}$, we denote by $|\vec{\rho}|$ the cardinal number of $\vec{\rho}$.  $\mP_{\mS(d)}$ denotes the power set of $\mS(d)$.


\section{Preliminary}


\subsection{Individual state imaginarity}

Recall that, in the resource theory of imaginarity, the free states are {\it real states}, i.e., states with a real density matrix $\langle m |\rho |n\rangle\in{\mathbb R}$ under a given reference basis $\{|m\ra\}_{m=0}^{d-1}$ of $\mathcal H$. The corresponding free operation is called {\it real operation}, which can be represented as $\mE(\rho) =\sum_j K_j\rho K_j^{\dag}$ with the
Kraus operators $K_j$s satisfy $\langle m|K_j|n\rangle\in{\mathbb R}$. A nonnegative function ${\mathcal I}$ on $\mS(d)$ is
called an imaginarity measure if it admits the following conditions (I1) to (I4)~\cite{HickeyGour2018,WuKD2021pra,WuKD2021prl}: 

(I1) Non-negativity: $\mathcal I (\rho) \geqslant 0$, and $\mathcal I (\rho) = 0$ for any real state $\rho$. 

(I2) Monotonicity: $\mathcal I (\mE(\rho))\leqslant  \mathcal I (\rho)$ whenever $\mE$ is a real operation. 

(I3) Probabilistic monotonicity: $\sum_j p_j{\mathcal I}(\rho_j)\leqslant {\mathcal I}(\rho)$ for any set of real operators $\{K_j\}$ satisfying $\sum_j K_j^{\dag}K_j=I$, where $p_j=\tr (K_j\rho K_j^{\dag})$, $\rho_j=\frac {1}{p_j}K_j\rho K_j^{\dag}$. 

(I4) Convexity: ${\mathcal I}(\sum_jp_j\rho_j)\leqslant \sum_jp_j{\mathcal I}(\rho_j)$ for any ensemble $\{p_j,\rho_j\}$. 
Note that (I3) and (I4) together imply (I2). In Ref.~\cite{Xue2021qip}, another condition was discussed, i.e., 

(I5) Additivity for direct sum states: ${\mathcal I}(p\rho_1 \oplus (1-p)\rho_2) = p{\mathcal I}(\rho_1)+(1-p){\mathcal I}(\rho_2)$,~$0<p<1$. (I3) and (I4) are equivalent to (I2) and (I5)~\cite{Xue2021qip}.

Several measures of imaginarity have been proposed in literature~\cite{HickeyGour2018,WuKD2021pra,WuKD2021pra,Kondra,Xu,Xu2}.
For an individual quantum state, robustness is typically defined as the minimum proportion of a noise state (or an arbitrary state) that must be mixed with the given state to transform it into a ``free state''~\cite{HickeyGour2018,WuKD2021pra}, i.e.,
\be
\mI_R(\rho)=\min_\tau\left\{s\geqslant0:\frac{\rho+s\tau}{1+s}\in\mR\right\},
\ee
where $\mR$ denotes the set of all real states under the reference basis. For any metric $\mC$ that is contractive under CPTP maps,
\beax
\mI_C(\rho)=\min\limits_{\sigma\in\mR}\mC(\rho,\sigma)
\eeax
is also an imaginarity measure~\cite{HickeyGour2018}. This includes the the trace distance measure of imaginarity~\cite{HickeyGour2018}, relative entropy of imaginarity~\cite{Xue2021qip}, and the geometric measure of imaginarity~\cite{Kondra}. In particular, the trace distance is such a case, which is defined by
\be
\mI_\tr(\rho)=\min\limits_{\sigma\in\mR}\|\rho-\sigma\|_{\tr}.
\ee
It was shown in Refs.~\cite{WuKD2021prl,WuKD2021pra} that $\mI_\tr$ coincides with $\mI_R$ indeed.

\subsection{Set imaginarity}

A set of states
\beax
\vec{\rho}=\{\rho_j: 1\leqslant j\leqslant n\}\subseteq \mS(d)
\eeax
is called \emph{set imaginarity free} if there exists a unitary operator $U$ acting on $\mH$ such that $U\rho_jU^\dag$ is real with respect to some basis $\{|i\rangle\}$ for all $j$~\cite{Miyazaki2022q,LiMS2026pra}. Otherwise, we say that the set $\vec{\rho}$ has set imaginarity. It was shown in Ref.~\cite{LiMS2026pra} that any pair of states $\{\rho_1, \rho_2\}\subseteq\mS(d)$ is set imaginarity free.
So we always assume throughout this paper that $n\geqslant3$ unless otherwise specified. We denote by $\check{\mR}_n$ the set of all set imaginarity free states $\vec{\rho}$ with $|\vec{\rho}|=n$. The set of all real states under a given reference basis is convex and affine in the sense that it is closed under all affine combinations that remain valid density operators~\cite{HickeyGour2018}, but $\check{\mR}_n$ is not convex in general~\cite{LiMS2026pra}.

Li \emph{et al.} defined the robustness of set imaginarity  $\vmI_{R}(\vec{\rho})$ by~\cite{LiMS2026pra}
\be
\vmI_{R}(\vec{\rho})=\min_U\frac{1}{n}\sum_{j=1}^{n}\mI_R(U\rho_jU^\dagger)
\ee
with respect to the given reference basis, where the minimization is taken over all unitary operators $U$.
Miyazaki and Matsumoto~\cite{Miyazaki2022q} defined mean-distance
\be
\vmI_C(\vec{\rho}) = \min\limits_{U} \frac{1}{n} \sum\limits_{j=1}^{n} \mI_C(U\rho_jU^\dag),
\ee
and max-distance
\be
\vmI'_C(\vec{\rho}) = \min\limits_{U} \max\limits_{j} \mI_C(U\rho_jU^\dag),
\ee
in the same manner.
When we take the trace distance in $\mI_C$, $\vmI_{R}$ is equivalent to $\vmI_C$.
Hereafter, for any set of qudit states $\vec{\rho}=\{\rho_1,\rho_2,\dots,\rho_n\}$, we write
\bea\label{I_n}
\vmI(\vec{\rho}) = |\rIm \Delta_n(\vec{\rho})|=|\rIm\tr(\rho_1\rho_2\cdots \rho_n)|
\eea
and
\beax
\mE(\vec{\rho})=(\mE(\rho_1), \mE(\rho_2),\cdots,\mE(\rho_n)),
\eeax
where $\mE$ is a quantum channel acting on $\mS(d)$.


\section{The free operation of set imaginarity for qubit system}


Recall that every qubit state $\rho\in\mS(2)$ admits the Bloch representation, i.e.,
\bea\label{bloch}
\rho=\frac{1}{2}(I+\bf{r}\cdot{\bm{\sigma}}),
\eea
where $\bm{\sigma}=(\sigma_{x}, \sigma_{y}, \sigma_{z})$ is the vector of Pauli operators, ${\bf{r}}=(r_x, r_y, r_z)$ is the Bloch vector of $\rho$, $\|\bf{r}\|\leqslant 1$.
Given a set of states $\vec{\rho}=\{\rho_j: 1\leqslant j\leqslant n\}\subseteq\mS(2)$, let $\bfr_i=(r_{x(i)}, r_{y(i)}, r_{z(i)})$ be the Bloch vector of $\rho_i\in\vec{\rho}$. Corollary III.7 in Ref.~\cite{LiMS2026pra} means that $\vec{\rho}$ lying within a plane defined by a great circle of the Bloch sphere is imaginarity free and vice versa.
For any quantum operation $\mE$ on $\vec{\rho}$, we write $\rho_i'=\mE(\rho_i)$ with the corresponding Bloch vector given by
\bea\label{vector}
{\bfr_i'}=T\bfr_i+\bf{c},
\eea
where $\bf{c}\in{\mathbb R}^3$ and $T$ is a $3\times3$ real matrix~\cite{Pasieka2009,HickeyGour2018}. It is clear that $\bf{c}=0$ if and only if $\mE$ is unital. Here we say a quantum channel $\mE$ is unital if $\mE(I)=I$. If $\mE\left[\mS(2)\right]$ is a plane that contains $I/2$, then $\mE(\vec{\rho})$ must be set imaginarity free. We show below that, these two kinds of operations are free operation of qubit set imaginarity and vice versa. For simplicity, we call $\mE$ a \emph{planarized operation} if $\mE\left[\mS(2)\right]$ is a plane that contains $I/2$ in the Bloch sphere. Namely, the term ``planarized operation'' refers to the fact that the image of the Bloch ball is contained in an affine subspace of dimension at most two. In particular, this includes the degenerate case where the image reduces to a line segment [$\operatorname{rank}(T)=1$] or a single point [$\operatorname{rank}(T)=0$].

\begin{theorem}\label{theorem1}
The free operation of qubit set imaginarity is either any common unital operation or any common planarized but nonunital operation.
\end{theorem}

\begin{proof}
In quantum resource theory, the free operation is the quantum channel that cannot convert free states into resource states~\cite{HickeyGour2018,Brandao2015prl}. It is also called resource nongenerating operation~\cite{HickeyGour2018,Brandao2015prl}. It is clear that, any common unitary operation is a free operation of qubit set imaginarity. We show below that free operation of qubit set imaginarity can also be other operations that contain unitary operation as a special case.

By Corollary III.7 in Ref.~\cite{LiMS2026pra}, we begin with the set consisting of three states, i.e., $\vec{\rho}=\{\rho_1,\rho_2,\rho_3\}$. It turns out that
\beax
\vmI(\vec{\rho})&=&|\rIm\tr(\rho_1\rho_2\rho_3)| =\frac{1}{4}|\det(\bfr_{1}, \bfr_{2}, \bfr_{3})|,\\
\vmI\left[\mE(\vec{\rho})\right]&=&|\rIm\tr(\rho_{1}'\rho_{2}'\rho_{3}')|=\frac{1}{4}|\det(\bfr_{1}', \bfr_{2}', \bfr_{3}')|,
\eeax
where $\bfr_{i}$ and $\bfr_{i}'$ denotes the Bloch vectors of $\rho_i$ and $\rho'_i$, respectively.
There are two different cases: (i) $\bf{c}=\bf{0}$ and (ii) $\bf{c}\neq\bf{0}$.

Case 1: $\bf{c}=\bf{0}$. In such a case, the center of the sphere is fixed under $\mE$. Obviously, $\|{\bf{r'}}\|=\|T{\bf{r}}+{\bf{c}}\|=\|T{\bf{r}}\|\le1$, so $\|T\|\le1$, and thus $|\!\det T|\le1$. It follows that
\beax
\left|\det(\bfr_{1}', \bfr_{2}', \bfr_{3}')\right|
=|\det T || \det(\bfr_{1}, \bfr_{2}, \bfr_{3})|
\le|\det(\bfr_{1}, \bfr_{2}, \bfr_{3})|,
\eeax
that is,
\bex
|\rIm\tr(\rho_{1}'\rho_{2}'\rho_{3}')|\leqslant|\rIm\tr(\rho_1\rho_2\rho_3)|.
\eex
Notice that $\bf{c}=\bf{0}$ if and only if $\mE$ is unital. We now conclude, based on Corollary III.7 in Ref.~\cite{LiMS2026pra}, that any common unital operation is a free operation of qubit set imaginarity.

Case 2: $\bf{c}\neq\bf{0}$. If $\operatorname{rank}(T)=3$, then we can choose $\bfu,\bfv\in\mathbb R^3$ such that
\bex
\det(\bfu,\bfv,\bfc)\neq 0.
\eex
Define
\bex
\bfr_1=\varepsilon T^{-1}\bfu,~
\bfr_2=\varepsilon T^{-1}\bfv,~
\bfr_3=\varepsilon(\alpha T^{-1}\bfu+\beta T^{-1}\bfv),
\eex
where $\alpha+\beta\neq 1$ and $\varepsilon>0$ is sufficiently small so that all three vectors lie in the Bloch ball. Then $\bfr_1,\bfr_2,\bfr_3$ are coplanar, but
\bex
\det(\bfr_1',\bfr_2',\bfr_3')=(1-\alpha-\beta)\varepsilon^2 \det(\bfu,\bfv,\bfc)\neq 0.
\eex
Namely, any nonunital operation with $\operatorname{rank}(T)=3$ is not a free operation of the qubit set imaginarity.

We now discuss the case of $\operatorname{rank}(T)= 2$. We denote by $R(T)$ the range of $T$. If $\bfc \in R(T)$, then we let $\bfr_1,\bfr_2,\bfr_3$ be any coplanar triple of Bloch vectors. With no loss of generality, we assume that
\[
\bfr_3 = \alpha \bfr_1 + \beta \bfr_2
~\text{for some } \alpha,\beta \in \mathbb{R}.
\]
Under the action of $\mE$,
\[
\bfr_i' = T\bfr_i+\bfc,\quad i=1,2,3,
\]
and hence
\[
\bfr_3'=\alpha T\bfr_1+\beta T\bfr_2+\bfc.
\]
It turns out that
\beax
&&\det(\bfr_1',\bfr_2',\bfr_3')\\
&=& \det(T\bfr_1+\bfc,\; T\bfr_2+\bfc,\; \alpha T\bfr_1+\beta T\bfr_2+\bfc)\\
&=&(1-\alpha-\beta)\det(T\bfr_1,T\bfr_2,\bfc)
=0.
\eeax
That is, every coplanar input triple is mapped to a coplanar output triple, and thus $\mE$ is a free operation.

If $\bfc\notin R(T)$, then we can choose linearly independent vectors $\bfu,\bfv\in R(T)$ such that
\[
\det(\bfu,\bfv,\bfc)\neq 0.
\]
Since $\bfu,\bfv \in R(T)$, there exist $\bfs_1,\bfs_2 \in \mathbb{R}^3$ such that
\[
T\bfs_1 = \bfu, ~ T\bfs_2 = \bfv.
\]
Define
\[
\bfr_1=\varepsilon \bfs_1,~
\bfr_2=\varepsilon \bfs_2,~
\bfr_3=\varepsilon(\alpha \bfs_1+\beta \bfs_2),
\]
with $\alpha+\beta\neq 1$ and $\varepsilon>0$ sufficiently small. Then the input triple is coplanar, but
\[
\det(\bfr_1',\bfr_2',\bfr_3')
=(1-\alpha-\beta)\varepsilon^2 \det(\bfu,\bfv,\bfc)\neq 0.
\]

If $\operatorname{rank}(T)\leqslant 1$, then in this case $\dim R(T)\leqslant 1$, and thus for any $\bfr_1,\bfr_2$, the vectors $T\bfr_1$ and $T\bfr_2$ are linearly dependent. Therefore,
\[
\det(T\bfr_1,T\bfr_2,\bfc)=0
\]
for all $\bfc$, and hence $\mE$ is a free operation.

In addition, if a set of free operations $\vec{\mE}=\{\mE_i\}$ acts on $\vec{\rho}$,
\bex
\vec{\mE}(\vec{\rho})=\left(\mE_1(\rho_1),\mE_2(\rho_2),\dots,\mE_n(\rho_n)\right),
\eex
we show below that $\vec{\mE}$ can create set imaginarity, namely, any set of free operations is not free operations any more.

Let $\vec{\rho}=\{\rho_1, \rho_2, \rho_3\}$ with the Bloch vectors are
\bex
\bfr_1=\frac1{\sqrt2}(1,0,-1), ~\bfr_2=\frac1{\sqrt2}(0,1,-1), ~\bfr_3=\frac1{\sqrt2}(1,-1,0),
\eex
respectively. Then $\det(\bfr_1,\bfr_2,\bfr_3)=0$, and thus
\bex
\vmI(\vec{\rho})=0.
\eex
Let $\mE_1$ be the phase-flip channel with $p=\frac{1}{2}$, and then
\bex
T = \begin{pmatrix}
	0 & 0 & 0 \\
	0 & 0 & 0 \\
	0 & 0 & 1
\end{pmatrix},
\quad \bf{c} = \begin{pmatrix}
	0 \\ 0 \\ 0
\end{pmatrix}
\eex
with respect to the Bloch representation. Taking
\bex
\mE_2=\mE_3=\mathbbm{1},
\eex
where $\mathbbm{1}$ denotes the identity operation, the output Bloch vectors are
\beax
\bfr_1'&=&\frac1{\sqrt2}(0,0,-1),\\
\bfr_2'&=&\frac1{\sqrt2}(0,1,-1), \\
\bfr_3'&=&\frac1{\sqrt2}(1,-1,0),
\eeax
respectively. We get
\beax
\vmI[\mE(\vec{\rho})]=|\rIm\tr(\rho_1' \rho_2' \rho_3')|
=\frac{\sqrt2}{16}>\vmI(\vec{\rho})=0.
\eeax
In general, for any $\vec{\rho}=\{\rho_1, \rho_2, \dots, \rho_n\}$ with $n\geq3$, we assume that it is set imaginarity free, and then they lie within a plane defined by a great circle of the Bloch sphere according to Ref.~\cite{LiMS2026pra}. We can always find different ``rotations'' induced by unital operations or planarized operations such that these points after rotating are not lying in a same plane any more. 
\end{proof}

The unital operations cover a wide range of qubit channels, since any unitary operation is unital, and in addition, such as the phase-flip channel $\mE(\rho)=(1-p)\rho+p\sigma_{z}\rho\sigma_{z}$, the bit-flip channel $\mE(\rho)=(1-p)\rho+p\sigma_{x}\rho\sigma_{x}$, and the depolarizing channel $\mE(\rho)=(1 - p)\rho + \frac{p}{3}(\sigma_x\rho\sigma_x + \sigma_y\rho\sigma_y + \sigma_z\rho\sigma_z)$ are also unital.


\section{The axiomatic definition of set-imaginarity measure}


Any pair of states $\{\rho_1, \rho_2\}$ is always set imaginarity free, whereas the enlarged set $\{\rho_1,\rho_2,\rho_3\}$ may exhibit set imaginarity. In general, for any sets $\vec{\rho}$ and $\vec{\varrho}$, if $\vec{\rho}\subseteq\vec{\varrho}$, the imaginarity is increasing intuitively when $\vec{\rho}$ becomes $\vec{\varrho}$ by adding states. In addition, if $\vec{\varrho}\cap\vec{\rho}=\emptyset$, the combined set $\vec{\varrho}\cup\vec{\rho}$ contains both $\vec{\varrho}$ and $\vec{\rho}$. As a resource, the `amount' of imaginarity in $\vec{\varrho}\cup\vec{\rho}$ should be not less than that of the sum of $\vec{\varrho}$ and $\vec{\rho}$. For example, if both $\vec{\rho}$ and $\vec{\sigma}$ are sets of qubit states and lie in two different planes that including the center of the Bloch sphere respectively, then both $\vec{\rho}$ and $\vec{\sigma}$ are set imaginarity free
	but	$\vec{\rho}\,\cup\,\vec{\sigma}$ is not. Moreover, if $\vec{\rho}\cap\vec{\sigma}=\emptyset$ additionally,  $\vec{\rho}\cap\vec{\sigma}$ contains no set imaginarity. For the high-dimensional case, if $\vec{\varrho}$ and $\vec{\rho}$ are two commutative sets of states, i.e., the states in $\vec{\rho}$ (respectively, $\vec{\sigma}$) are commutative with each other, but there exists $\rho_i\in\vec{\rho}$ and $\sigma_j\in\vec{\sigma}$ such that $[\rho_i, \sigma_j]\neq0$, then both $\vec{\rho}$ and $\vec{\sigma}$ are set imaginarity free but it is possible that
	$\vec{\rho}\,\cup\,\vec{\sigma}$ contains set imaginarity. Moreover, if $\vec{\rho}\cap\vec{\sigma}=\emptyset$ or $\vec{\rho}\cap\vec{\sigma}$ is a commutative set additionally, $\vec{\rho}\cap\vec{\sigma}$ contains no set imaginarity. In addition, by Theorem~\ref{theorem1}, the free operation of set imaginarity should be a common single operation. With these observations keep in mind, we now give the definition of the SIM as follows.

\begin{definition}\label{SIM}
A function $\vmI: \mP_{\mS(d)}\mapsto\mbR_+$ is called a SIM if it satisfies the following conditions (SI1)--(SI3):

(SI1) $\vmI(\vec{\rho})=0$ if $\vec{\rho}$ is set imaginarity free;


(SI2) For any unitary $U$,
\be 
\vmI(\vec{\rho})
=\vmI(U\vec{\rho} U^\dagger)
\ee
for any $\vec{\rho}$;

(SI3) $\vmI$ is monotonic under any free operation of set imaginarity, i.e.,
\be
\vmI(\vec{\rho})\geqslant \vmI\left[ \mE(\vec{\rho})\right]
\ee
for any $\vec{\rho}$ and any free operation of set imaginarity $\mE$.

Let $\vmI$ be a SIM. If
(SI4) for any $\vec{\rho}\subseteq\vec{\varrho}$, 
\be 
\vmI(\vec{\rho})\leqslant \vmI(\vec{\varrho}),
\ee
we call it a unified SIM. Moreover, for a unified SIM $\vmI$ if (SI5) $\vec{\rho}\,\cap\,\vec{\sigma}=\emptyset$, then \be 
\vmI(\vec{\rho})+\vmI(\vec{\sigma})\leqslant \vmI(\vec{\rho}\,\cup\,\vec{\sigma}),
\ee we call it a complete SIM.  Here we let $\vmI(\emptyset)\equiv0$ and $\vmI(\rho)\equiv0$ for any single state $\rho$. In addition, if a unified SIM $\vmI$ satisfies
(SI6) 
\be \vmI(\vec{\rho})+\vmI(\vec{\sigma})\leqslant \vmI(\vec{\rho}\cup\vec{\sigma})+\vmI(\vec{\rho}\cap\vec{\sigma}),
\ee 
it is said to be strongly superadditive.
\end{definition}

If $\vmI$ admits item (SI4), then it is said to be set monotonic. Item (SI5) is called superadditivity and item (SI6) is called strong superadditivity.

\begin{table}
	\caption{\label{tab:table0} Comparison of unified and complete SIMs with unified and complete MEMs, respectively.
        Free operation of set imaginarity is abbreviated to FOoSI.}
	\begin{ruledtabular}
		\begin{tabular}{cc}
			Coarsening relation/compatible & Set relation/compatible\\
			resource measure & set-resource measure \\ \colrule
			$\succ^c$/EM & $\vec{\rho}\xrightarrow{\mE} \vec{\rho'}$, \text{FOoSI}~$\mE$/SIM \\
			$\succ^a$/unified MEM & $\vec{\varrho}\rightarrow\vec{\rho}$, $\vec{\rho}\subsetneq\vec{\varrho}$/unified SIM\\
			$\succ^b$/complete MEM & $\vec{\varrho}\rightarrow\left\lbrace \vec{\rho}, \vec{\sigma}\right\rbrace$,\\            &$\vec{\varrho}=\vec{\rho}\cup\vec{\sigma}$/complete SIM
		\end{tabular}
	\end{ruledtabular}
\end{table}

\begin{remark}
For any $\vec{\rho}\subset\mS(d)$ and $\vec{\varrho}\subset\mS(d')$ with $|\vec{\rho}|=|\vec{\varrho}|$, we let
\beax
p\vec{\rho} \oplus (1-p)\vec{\varrho}&:=&\{p\rho_1\oplus(1-p)\varrho_1, p\rho_2\oplus(1-p)\varrho_2,\\
&&~\cdots, p\rho_n\oplus(1-p)\varrho_n\}.
\eeax
The example in Ref.~\cite{LiMS2026pra} shows that, for any SIM $\vmI$, $\vmI\left(p\vec{\rho} \oplus (1-p)\vec{\varrho} \right)\neq p\vmI(\vec{\rho}) + (1-p)\vmI(\vec{\varrho})$ for some $\vec{\rho} $, $\vec{\varrho}$, and $0<p<1$. Namely, there is no counterpart of direct-sum additivity for SIM [i.e., condition (I5)].
Note that there is no counterpart of probabilistic monotonicity for SIMs either.
In fact, for any free operation of set imaginarity with Kraus operators $\{K_j\}$, we let $ p_{ij} = \tr(K_j \rho_i K_j^\dagger)$, $\vec{\rho}_{(j)}=\{\rho_{1j}, \rho_{2j}, \dots, \rho_{nj}\}$, $\rho_{ij} = \frac{1}{p_{ij}} K_j \rho_i K_j^\dagger $, $1\leqslant i\leqslant n$, and then $p_{ij}\neq p_{i'j}$ in general whenever $i\neq i'$. Only in the special case where, for each outcome $j$, the probability $p_{ij}$ is independent of the input state $\rho_i$, i.e., $p_{ij}=p_{i'j}$ for all $1\leqslant i, i'\leqslant n$, can we write
$p_{ij}=p_j$ and impose the probabilistic monotonicity condition $\sum\limits_jp_j\vmI(\vec{\rho}_{(j)})\leqslant\vmI(\vec{\rho})$. The counterpart of convexity for SIMs is also meaningless since there are no simultaneously convex combinations $\rho_j=\sum\limits_{i}p_i\rho_{ij}$ for all $j$ in general.
\end{remark}

\begin{remark}
The terms unified SIM and complete SIM here are used in analogy with the unified multipartite entanglement measure (MEM) and complete MEM introduced in Refs.~\cite{Guo2020pra,Guo2022entropy,Guo2024pra,Guo2024rip,Guo2025arxiv}.
A detailed review of these concepts is provided in Appendix~\ref{a} and Appendix~\ref{b}.
Coarsening relation within system partitions serves as the primary method for MEMs classification. There are three basic coarsening relations, i.e., coarsening relations of type (a), (b), and (c), respectively. The original MEM is mainly based on coarsening relation of type (a), the unified MEM is based on type (b), while the complete MEM corresponds to type (c). The counterparts of these three basic coarsening relations for the state-set resources,
in a somewhat imprecise sense, can be regarded as the set of states transformation under free operation, discarding some of the states in the set, and dividing the set of states into subsets, respectively. The comparison is summarized in Table~\ref{tab:table0}.
\end{remark}


\section{Qubit set-imaginarity measure via the Bargmann invariants}


In this paper, we concentrate on the qubit set-imaginarity measure (QbSIM), because the precise form of the free operation for qudit set imaginarity beyond qubits ($d\geqslant3$) remains unknown. We begin with discussing the quantity $\vmI$ in Eq.~\eqref{I_n}. We let  $\vec{\varrho}=\{\rho_1, \rho_2, \rho_3, \rho_4\}$ with the Bloch vectors ${\bfr_1}=(1, 0, 0)$, ${\bfr_2}=(0, 1, 0)$, ${\bfr_3}=(0, 0, 1)$,  ${\bfr_4}=(-\frac12, 0, -\frac12)$, respectively. Then
$\vmI(\vec{\varrho})=\left| \rIm \tr(\rho_1 \rho_2 \rho_3 \rho_4)\right|=\frac{1}{8}\left| \det(\bfr_1, \bfr_2, \bfr_3)\right.+ \det(\bfr_1, \bfr_2, \bfr_4)+ \det(\bfr_1, \bfr_3, \bfr_4) \left. +\det(\bfr_2, \bfr_3, \bfr_4)\right|=0$. Taking $\vec{\rho}=\{\rho_1, \rho_2, \rho_3\}$, we have $\vmI(\vec{\rho})=|\rIm\tr(\rho_1\rho_2\rho_3)| =\frac{1}{4}|\det(\bfr_{1}, \bfr_{2}, \bfr_{3})|=\frac14$; that is $\vmI(\vec{\rho}) > \vmI(\vec{\varrho})=0$. We thus conclude that Eq.~\eqref{I_n} is unable to be a proper quantity of the set imaginarity. Together with Corollary III.7 in~Ref.~\cite{LiMS2026pra}, we give the following quantity.

\begin{definition}
For a set of $n$ qubit states $\vec{\rho}=\{\rho_j: 1\leqslant j\leqslant n\}$, $n\geqslant3$, we define
\bea
\vmI_{\max}(\vec{\rho}) = \max_{1 \leqslant i < j < k \leqslant n}\left|\rIm\tr(\rho_i \rho_j \rho_k)\right|.
\eea
\end{definition}

It is clear that the items $\mathrm{(SI1)}$--$\mathrm{(SI4)}$ in Definition~\ref{SIM} hold true for $\vmI_{\max}$. It is worth mentioning here that for any SIM $\vmI$ and $\vec{\rho}$, $\vmI(\vec{\rho})$ is defined on the set $\vec{\rho}$. So $\vmI(\vec{\rho})$ is invariant under permutation of the states in $\vec{\rho}$ although we always assume that there is a order in $\vec{\rho}$ when we discuss a concrete SIM, i.e., for example,
\bex
\vmI_{\max}(\{\rho_1,\rho_2, \dots, \rho_n\})=\vmI_{\max}(\{\rho_{\pi(1)},\rho_{\pi(2)}, \dots, \rho_{\pi(n)}\})
\eex
for any permutation $\pi$ acting on $\{1, 2, \dots, n\}$, which is clear by the conjugation symmetry of Bargmann invariants~\cite{TYchien2016}
\bex
|\rIm\tr(\rho_{1}\rho_{2}\cdots\rho_{n})|
=
|\rIm\tr(\rho_{n}\rho_{n-1}\cdots\rho_{1})|.
\eex
Moreover, we have the following proposition.

\begin{pro}\label{pro1}
$\vmI_{\max}$ is a unified QbSIM but it is not complete.
\end{pro}

\begin{proof}
We now give counterexamples which show that $\vmI_{\max}$ violates items (SI5) and (SI6) in Definition~\ref{SIM}.
We take $\vec{\rho}=\{\rho_1, \rho_2, \rho_3\}$ and $\vec{\sigma}=\{\rho_4, \rho_5, \rho_6\}$ with the corresponding Bloch vectors $\bfr_{1}=(1, 0, 0)$, $\bfr_{2}=(0, 1, 0)$, $\bfr_{3}=(0, 0, 1)$, $\bfr_{4}=(\frac{1}{2}, 0, 0)$, $\bfr_{5}=(0, \frac{1}{2}, 0)$, and $\bfr_{6}=(0, 0, \frac{1}{2})$, respectively. Then
\begin{align*}
	&\vmI_{\max}(\vec{\rho})=\frac{1}{4}, \quad
	\vmI_{\max}(\vec{\sigma})=\frac{1}{32}, \\
	&\vmI_{\max}(\vec{\rho} \cup \vec{\sigma})=\frac{1}{4}, \quad
	\vmI_{\max}(\vec{\rho} \cap \vec{\sigma})=0.
\end{align*}
Therefore,
\bex
\vmI_{\max}(\vec{\rho})+\vmI_{\max}(\vec{\sigma})\geqslant \vmI_{\max}(\vec{\rho} \cup \vec{\sigma}).
\eex
That is, $\vmI_{\max}$ violates item (SI5).

Let $\vec{\rho}=\{\rho_1, \rho_2, \rho_3, \rho_7, \rho_8, \rho_9\}$ and $\vec{\sigma}=\{\rho_4, \rho_5, \rho_6, \rho_7, \rho_8, \rho_9\}$ with the corresponding Bloch vectors are $\bfr_{1}=(1, 0, 0)$, $\bfr_{2}=(0, 1, 0)$, $\bfr_{3}=(0, 0, 1 )$, $\bfr_{4}=(\frac{1}{2}, 0, 0 )$, $\bfr_{5}=(0, \frac{1}{2}, 0  )$, $\bfr_{6}=(0, 0, \frac{1}{2})$, $\bfr_{7}=(\frac14, 0, 0)$, $\bfr_{8}=(0, \frac14, 0)$, and $\bfr_{9}=(0, 0, \frac14)$, respectively. We get
\begin{align*}
	&\vmI_{\max}(\vec{\rho})=\frac{1}{4}, \quad
	\vmI_{\max}(\vec{\sigma})=\frac{1}{32}, \\
	&\vmI_{\max}(\vec{\rho} \cup \vec{\sigma})=\frac{1}{4}, \quad
	\vmI_{\max}(\vec{\rho} \cap \vec{\sigma})=\frac{1}{256}.
\end{align*}
Therefore,
\bex
\vmI_{\max}(\vec{\rho})+\vmI_{\max}(\vec{\sigma})\geqslant \vmI_{\max}(\vec{\rho} \cup \vec{\sigma})+\vmI_{\max}(\vec{\rho} \cap \vec{\sigma}).
\eex
Hence, items~(SI5) and (SI6) fail.
\end{proof}

For any triplet $\{\rho_1,\rho_2,\rho_3\}$, $\bigl|\rIm\tr(\rho_1 \rho_2 \rho_3)\bigr|= \frac14 \bigl|\bfr_1 \cdot (\bfr_2 \times \bfr_3)\bigr|\leqslant \frac14 \|\bfr_1\|\,\|\bfr_2\|\,\|\bfr_3\|\leqslant \frac14$, that is, for any $\vec{\rho}$ with $|\vec{\rho}|=n \geqslant 3$,
\bea
\vmI_{\max}(\vec{\rho}) \leqslant \frac14.
\eea
The upper bound is achieved by choosing $\bfr_1=(1,0,0)$, $\bfr_2=(0,1,0)$, and $\bfr_3=(0,0,1)$, for which the corresponding states are
$\rho_1 = |+\rangle\langle +|$,~$\rho_2 = |\hat{+}\rangle\langle \hat{+}|$, and~$\rho_3 = |0\rangle\langle 0|$, respectively,
where $|+\rangle = (|0\rangle+|1\rangle)/\sqrt2$ and $|\hat{+}\rangle = (|0\rangle+\rmi|1\rangle)/\sqrt2$.

From the arguments above, we give another quantity based on the Bargmann invariants of the triplet states. For a set of qubit states $\vec{\rho} = \{\rho_j : 1 \leqslant j \leqslant n\}$, $n\geqslant3$, we let
\be \label{sum}
\vmI_{\Sigma}(\vec{\rho})=\sum_{1 \leqslant i < j < k \leqslant n} | \rIm \tr(\rho_i \rho_j \rho_k) |.
\ee

\begin{theorem}\label{theorem2}
$\vmI_{\Sigma}$ is a complete QbSIM and it is strongly superadditive.
\end{theorem}

\begin{proof}
By definition, together with Proposition~\ref{pro1}, it is clear that the items (SI1)--(SI5) in Definition~\ref{SIM} hold true for $\vmI_{\Sigma}$.

We now check the item (SI6). We let $\vec{\rho}=\{\rho_1, \rho_2, \dots, \rho_p\}$ and $\vec{\sigma}=\{\sigma_1, \sigma_2, \dots, \sigma_q\}$ be two finite sets of qubit states, and denote by $\vec{\omega}=\vec{\rho}\cup\vec{\sigma}$. If $\left| \vec{\rho}\cap\vec{\sigma}\right|\leqslant 3$, then the strong superadditivity is clear by definition. We assume that $\rho_i=\sigma_i$ for $1\leqslant i\leqslant s$, $3<s<p$, $s<q$. It follows that
\begin{widetext}
	\beax \vmI_{\Sigma}(\vec{\rho})&=&\sum_{1 \leqslant i < j < k \leqslant s} | \rIm \tr(\rho_i \rho_j \rho_k) |
	+\sum_{1 \leqslant i < j  \leqslant s<k} | \rIm \tr(\rho_i \rho_j \rho_k) |
	+\sum_{1 \leqslant i \leqslant s<j<k} | \rIm \tr(\rho_i \rho_j \rho_k) |
	+\sum_{s < i < j <k} | \rIm \tr(\rho_i \rho_j \rho_k) |,\\
	\vmI_{\Sigma}(\vec{\sigma})&=&\sum_{1 \leqslant i < j < k \leqslant s} | \rIm \tr(\sigma_i \sigma_j \sigma_k) |
	+\sum_{1 \leqslant i < j  \leqslant s<k} | \rIm \tr(\sigma_i \sigma_j \sigma_k) |
	+\sum_{1 \leqslant i \leqslant s<j<k} | \rIm \tr(\sigma_i \sigma_j \sigma_k) |
	+\sum_{s < i < j <k} | \rIm \tr(\sigma_i \sigma_j \sigma_k) |,
	\eeax
	and
	\beax \vmI_{\Sigma}(\vec{\rho}\cap\vec{\sigma})=\sum_{1 \leqslant i < j < k \leqslant s} | \rIm \tr(\sigma_i \sigma_j \sigma_k) |
	=\sum_{1 \leqslant i < j < k \leqslant s} | \rIm \tr(\rho_i \rho_j \rho_k) |.\eeax
\end{widetext}
This implies that item (SI6) holds obviously.
\end{proof}

Let $\vec{\rho}=\{\rho_i\}_{i=1}^4$ be a set of four qubit states with Bloch vectors are $\{\bfr_i\}_{i=1}^4$. Then
\bea
\check I_\Sigma(\vec{\rho})
\le
\frac{4}{3\sqrt{3}}.
\eea
Moreover, the equality holds if and only if there exist signs $\varepsilon_i\in\{\pm1\}$ such that $\{\varepsilon_i \bfr_i\}_{i=1}^4$
form a regular tetrahedron in the Bloch sphere, i.e.,
\[
\|\bfr_i\|=1,~
(\varepsilon_i \bfr_i)\cdot(\varepsilon_j \bfr_j)=-\frac{1}{3},~~ i\ne j.
\]

In order to see this, we let
\[
S_4:=\sum_{1\leqslant i<j<k\leqslant 4} |\det(\bfr_i,\bfr_j,\bfr_k)|,
\]
and $D_1=\det(\bfr_2,\bfr_3,\bfr_4)$,
$D_2=\det(\bfr_1,\bfr_3,\bfr_4)$,
$D_3=\det(\bfr_1,\bfr_2,\bfr_4)$,
$D_4=\det(\bfr_1,\bfr_2,\bfr_3)$,
$\bfx_1=\operatorname{sgn}(D_1)\bfr_1$,
$\bfx_2=-\operatorname{sgn}(D_2)\bfr_2$,
$\bfx_3=\operatorname{sgn}(D_3)\bfr_3$,
and $\bfx_4=-\operatorname{sgn}(D_4)\bfr_4$.
Then
\[
a_1 \bfx_1 + a_2 \bfx_2 + a_3 \bfx_3 + a_4 \bfx_4 = 0,
~a_i=|D_i|.
\]
Hence $0 \in \operatorname{conv}\{\bfx_1,\bfx_2,\bfx_3,\bfx_4\}$. Let $T=\operatorname{conv}\{\bfx_1$, $\bfx_2$, $\bfx_3$, $\bfx_4\}$;
then $\operatorname{Vol}(T)=\frac{1}{6} S_4$ and $S_4=6\,\operatorname{Vol}(T)$. Maximizing under $\|\bfx_i\|\leqslant 1$, the maximum occurs on the sphere and is achieved by the regular tetrahedron:
\[
\operatorname{Vol}_{\max}=\frac{8}{9\sqrt3}.
\]

The case $n=4$ is special in that the geometric structure becomes rigid: The extremal configuration corresponds (up to sign changes) to a regular tetrahedron, i.e., the qubit symmetric informationally complete positive operator-valued measure~\cite{Renes2004}.
For larger sets, the situation appears to be substantially more complicated. In particular, for $n\geqslant 5$, the quantity
$\check I_\Sigma(\vec{\rho})$ no longer reduces to a simple extremal polytope problem with a unique highly symmetric solution.
Preliminary numerical evidence suggests that the optimal configurations for $n=5$ are not given by standard spherical codes such as the triangular bipyramid, and exhibit a more intricate structure. Determining the exact maximum of $\check I_\Sigma$ for $n\geqslant 5$,
as well as characterizing the corresponding optimal configurations, remains an interesting open problem.


\section{Robustness of qubit set imaginarity}


Although $\vmI_{R}$ was first introduced in Ref.~\cite{LiMS2026pra}, its status as a rigorously defined SIM remains an open question.
Below, we show that for qubits it behaves as a QbSIM yet still falls short of the finer criteria needed to fully quantify  set imaginarity in a fine-grained manner.

\begin{pro}\label{pro2}
For the qubit system, $\vmI_{R}$ is a SIM, but it is neither unified nor complete.
\end{pro}

\begin{proof} We only need to check items (SI3)--(SI5) in Definition~\ref{SIM}. According to Ref.~\cite{LiMS2026pra}, $\vmI_{R}$ can be calculated as
	\be
	\vmI_{R}(\vec{\rho})
	=
	\min_{\bfp\in S^2}\frac1n\sum_{j=1}^n |\langle \bfr_j|\bfp\rangle|,
	\ee
	where $S^2=\{\bfp\in\mathbb{R}^3:\|\bfp\|=1\}$ and the minimum is taken over all unit vectors on $S^2$. For any set $\vec{\rho}$ and any unital channel $\mE$, we have
	\bex
	\vmI_{R}[\mE(\vec{\rho})]
	=\min_{\bfp\in S^2}\frac1n\sum_{j=1}^n |\langle \bfr_jT^\dag|\bfp\rangle|,
	\eex
	where $T$ is the associated matrix as in Eq.~(11). If $\rank(T)\le2$, then we have $\vmI_{R}[\mE(\vec{\rho})]=0$. If $\rank(T)=3$, for any $\bfp\in S^2$, then we have $T^\dagger\bfp \neq 0$, and $\|T^\dagger\bfp\| \leqslant 1$. It turns out that
	\beax
	\vmI_{R}[\mE(\vec{\rho})]
	&=&\min_{\bfp\in S^2}\frac1n\sum_{j=1}^n |\langle \bfr_jT^\dag|\bfp\rangle|\\
	&=&\min_{\bfp\in S^2}\frac1n\sum_{j=1}^n |\langle \bfr_j|T^\dagger \bfp\rangle|\\
	&\leqslant&\min_{\bfp\in S^2}\frac1n\sum_{j=1}^n |\langle \bfr_j|\bfp\rangle|\\
	&=&\vmI_{R}(\vec{\rho}),
	\eeax
	as desired. If $\mE$ is a planarized channel, then
	\bex
	\vmI_{R}[\mE(\vec{\rho})]
	=\min_{\bfp\in S^2}\frac1n\sum_{j=1}^n |\langle \bfr_jT^\dag+\bfc|\bfp\rangle|=0,
	\eex
	since $\dim {\rm span}\{T\bfr_j+\bfc\}\leqslant 2$. Namely $\vmI_{R}(\vec{\rho}) \geqslant \vmI_{R}[\mE(\vec{\rho})]$ for planarized channel $\mE$.
	
	Consider $\vec{\rho}=\{\rho_1,\rho_2,\rho_3\}$ with Bloch vectors $\bfr_1=(1,0,0)$, $\bfr_2=(0,1,0)$, and $\bfr_3=(0,0,1)$, respectively. Then
	\bex
	\vmI_{R}(\vec{\rho})=\min_{\bfp\in S^2}\frac13(|p_x|+|p_y|+|p_z|).
	\eex
	For every unit vector $\bfp=(p_x,p_y,p_z)$,
	\bea\label{eq1}
	|p_x|+|p_y|+|p_z|\geqslant \sqrt{p_x^2+p_y^2+p_z^2}=1.
	\eea
	We take $\bfp=(1, 0, 0)$, $|p_x|+|p_y|+|p_z|=1$. Hence $\vmI_{R}(\vec{\rho})=\frac13$. Now let $\tau$ be the maximally mixed state, and let $\vec{\varrho}=\vec{\rho}\cup {\tau}$. Then $\vec{\rho}\subseteq \vec{\varrho}$. Taking $\bfp=(1, 0, 0)$ gives
	\beax
	\vmI_{R}(\vec{\varrho})
	&\leqslant&
	\frac14\bigl(
	|\langle \bfr_1|\bfp\rangle| + |\langle \bfr_2|\bfp\rangle| + |\langle \bfr_3|\bfp\rangle|
	+|\langle \bfr_{\tau}|\bfp\rangle| \bigr)\\
	&=&\frac{1}{4}(1+0+0+0)=\frac{1}{4}.
	\eeax
	This leads to $\vmI_{R}(\vec{\varrho})\leqslant \frac14 < \frac13 = \vmI_{R}(\vec{\rho})$. Namely $\vmI_{R}$ is not unified.

	Let $\vec{\rho}=\{\rho_1,\rho_2,\rho_3\}$ and $\vec{\sigma}=\{\rho_4,\rho_5,\rho_6\}$ with the Bloch vectors are $\bfr_1=(1,0,0)$, $\bfr_2=(0,1,0)$, $\bfr_3=(0,0,1)$, $\bfr_4=(0,0,0)$, $\bfr_5=(\frac12,0,0)$, and $\bfr_6=(0,0,\frac12)$, respectively. Clearly, $\vec{\rho}\cap \vec{\sigma}=\emptyset$. Since all Bloch vectors in $\vec{\sigma}$ lie in the plane $y=0$, $\vmI_{R}(\vec{\sigma})=0$. We now consider
	$\vec{\rho}\cup \vec{\sigma}=\{\rho_1,\rho_2,\rho_3,\rho_4,\rho_5,\rho_6\}$. Taking again $\bfp=(0, 1, 0)$, we obtain
	$\langle \bfr_1|\bfp\rangle=0$, $|\langle \bfr_2|\bfp\rangle|=1$, and $\langle \bfr_3|\bfp\rangle=0$.
	Hence $\vmI_{R}(\vec{\rho}\cup \vec{\sigma}) \le	\frac16(0+1+0+0+0+0)=\frac16 \leqslant\vmI_{R}(\vec{\rho})+\vmI_{R}(\vec{\sigma})=\frac13+0=\frac13$. That is, $\vmI_{R}$ is not superadditive in general.	
\end{proof}

$\vmI_{R}$ is, of course, not strongly superadditive. Taking $\vec{\rho}=\{\rho_1,\rho_2,\rho_3\}$ and $\vec{\sigma}=\{\rho_1,\rho_4,\rho_5\}$ with the Bloch vectors are $r_1=(1,0,0)$, $r_2=(0,1,0)$, $r_3=(0,0,1)$, $r_4=(0,0,0)$, and  $r_5=(\frac12,0,0)$, respectively, we have
$\langle \bfr_1|\bfp\rangle=0$, $\langle \bfr_2|\bfp\rangle=1$, $\langle \bfr_3|\bfp\rangle=0$, $\langle \bfr_4|\bfp\rangle=0$, and	$\langle \bfr_5|\bfp\rangle=0$ whenever $\bfp=(0, 1, 0)$. We thus get
\bex
\vmI_{R}(\vec{\rho})+\vmI_{R}(\vec{\sigma})=\frac13>  \frac15\geqslant\vmI_{R}(\vec{\rho}\cup \vec{\sigma})+\vmI_{R}(\vec{\rho}\cap \vec{\sigma}).
\eex

Analogously to that of the maximal robustness of set-coherence defined in Ref.~\cite{Designolle2021prl}, we can define
\be
\vmI'_{R}(\vec{\rho})=\min_U\max\limits_j\mI_R(U\rho_jU^\dagger),~\forall\,\vec{\rho}\subset\mS(d)
\ee
with respect to any given reference basis.

\begin{pro}
For qubit system, $\vmI'_{R}$ is a unified SIM, but it is neither complete nor strongly superadditive.
\end{pro}

\begin{proof}
	By definition, it is clear that the items (SI1) and (SI2) in Definition~\ref{SIM} hold true for $\vmI'_{R}$, we only need to check items (SI3)--(SI6).
	$\vmI'_{R}(\vec{\rho})$ can be equivalently written as
	\be
	\vmI'_{R}(\vec\rho)
	=\min_{\bfp\in S^2}\max_j |\langle \bfr_j|\bfp\rangle|,
	\ee
	where $S^2=\{\bfp\in\mathbb{R}^3:\|\bfp\|=1\}$ and the minimum is taken over all unit vectors on $S^2$.
	
	Let $\mE$ be any given unital channel. If $\rank(T)\leqslant 2$, then we have
	\bex
	\vmI'_{R}[\mE(\vec{\rho})]
	=\min_{\bfp\in S^2}\max_j |\langle \bfr_jT^\dag|\bfp\rangle|
	=0.
	\eex
	If $\rank(T)=3$, then $T^\dagger\bfp \neq 0$ and $\|T^\dagger\bfp\| \leqslant 1$. It turns out that
	\beax
	\vmI'_{R}[\mE(\vec\rho)]
	&=&\min_{\bfp\in S^2}\max_j |\langle \bfr_jT^\dag|\bfp\rangle|\\
	&=&\min_{\bfp\in S^2}\max_j |\langle \bfr_j|T^\dag\bfp\rangle|\\
	&\leqslant&\min_{\bfq\in S^2}\max_j|\langle \bfr_j|\bfp \rangle|\\
	&=&\vmI'_{R}(\vec\rho).
	\eeax
	If $\mE$ is a planarized channel, then
	\bex
	\vmI'_{R}[\mE(\vec{\rho})]
	=\min_{\bfp\in S^2}\max_j |\langle \bfr_jT^\dag+\bfc|\bfp\rangle|=0,
	\eex
	since $\dim {\rm span}\{T\bfr_j+\bfc\}\leqslant 2$. Namely $\vmI'_{R}(\vec\rho) \geqslant \vmI'_{R}[\mE(\vec{\rho})]$ for planarized channel $\mE$.

	Suppose that $\vec\rho\subseteq\vec\varrho$. Then for every $\bfp\in S^2$,
	\bex
	\max_{\rho_j\in\vec\rho}|\langle \bfr_j|\bfp\rangle|
	\le
	\max_{\rho_j\in\vec\varrho}|\langle \bfr_j|\bfp\rangle|.
	\eex
	Taking the minimum over $\bfp\in S^2$, we have
	\bex
	\vmI'_{R}(\vec\rho)\leqslant \vmI'_{R}(\vec\varrho).
	\eex

	Let $\vec{\rho}=\{\rho_1,\rho_2,\rho_3\}$ and $\vec{\sigma}=\{\rho_4,\rho_5,\rho_6\}$ with the Bloch vectors are $\bfr_1=(1,0,0)$, $\bfr_2=(0,1,0)$, $\bfr_3=(0,0,1)$, $\bfr_4=(\frac12,0,0)$, $\bfr_5=(0,\frac12,0)$, and $\bfr_6=(0,0,\frac12)$, respectively. Clearly, $\vec{\rho}\cap \vec{\sigma}=\emptyset$. For $\bfp=(p_x,p_y,p_z)\in S^2$, one has
	\bex
	\max_{1\leqslant j\leqslant 3}|\langle \bfr_j|\bfp\rangle|
	=\max\{|p_x|,|p_y|,|p_z|\}.
	\eex
	Since
	\[
	\max\{|p_x|, |p_y|, |p_z|\} \geqslant\sqrt{\frac{p_x^2 + p_y^2 + p_z^2}{3}} = \frac{1}{\sqrt{3}},
	\]
	taking $\bfp = \frac{1}{\sqrt{3}}(1, 1, 1)$, we obtain
	\bex
	\vmI'_{R}(\vec\rho)=\frac1{\sqrt3}, ~ \vmI'_{R}(\vec\sigma)=\frac1{2\sqrt3}.
	\eex
	Now consider
	$\vec{\rho}\cup \vec{\sigma}=\{\rho_1,\rho_2,\rho_3,\rho_4,\rho_5,\rho_6\}$, we find
	\bex
	\vmI'_{R}(\vec\rho\cup\vec\sigma)
	=\frac1{\sqrt3}.
	\eex
	Consequently,
	\bex
	\vmI'_{R}(\vec\rho)+\vmI'_{R}(\vec\sigma)
	=\frac1{\sqrt3}+\frac1{2\sqrt3}
	>\frac1{\sqrt3}
	=\vmI'_{R}(\vec\rho\cup\vec\sigma).
	\eex
	
	Let $\vec{\rho}=\{\rho_1,\rho_2,\rho_3\}$ and $\vec{\sigma}=\{\rho_1,\rho_4,\rho_5\}$ with the Bloch vectors are $\bfr_1=(1,0,0)$, $\bfr_2=(0,1,0)$, $\bfr_3=(0,0,1)$, $\bfr_4=(0,\frac12,0)$, and $\bfr_5=(0,0,\frac12)$, respectively. Obviously $\vmI'_{R}(\vec\rho\cap\vec\sigma)=0$.
	Taking $\bfp=\left(\frac13,\frac23,\frac23\right)$, we get
	\bex
	\vmI'_{R}(\vec\sigma)=\frac13.
	\eex
	Now consider $\vec{\rho}\cup \vec{\sigma}=\{\rho_1,\rho_2,\rho_3,\rho_4,\rho_5\}$,
	\beax
	&&\max_j |\langle \bfr_j,p\rangle|\\
	&=&\max\left\{ |p_x|, |p_y|, |p_z|, \frac{1}{2}|p_y|, \frac{1}{2}|p_z| \right\}\\
	&=&\max\left\{ |p_x|, |p_y|, |p_z| \right\}.
	\eeax
	Therefore, $\vmI'_{R}(\vec\rho\cup\vec\sigma)=\frac1{\sqrt3}$, and
	\[
	\vmI'_{R}(\vec\rho)+\vmI'_{R}(\vec\sigma)
	=\frac1{\sqrt3}+\frac13
	>\frac1{\sqrt3}
	=\vmI'_{R}(\vec\rho\cup\vec\sigma)	+\vmI'_{R}(\vec\rho\cap\vec\sigma).
	\]
	We thus proved that $\vmI'_{R}$ satisfies items (SI3) and (SI4) but violates items (SI5) and (SI6).	
\end{proof}

By the arguments above, the multiplier factor ``$1/n$'' in the expression of $\vmI_{R}$ seems the main cause of the violation of items (SI4)--(SI6).  We thus improve $\vmI_{R}$ as follows.

\begin{definition}
For any $\vec{\rho}\subset\mS(d)$, we define
\be
\vmI''_R(\vec{\rho})=\min_U\sum_{j=1}^{n}\mI_R(U\rho_jU^\dagger)
\ee
with respect to any given reference basis.
\end{definition}

\begin{pro}
For the qubit system, $\vmI''_R$ is a complete SIM but it is not strongly superadditive.
\end{pro}

\begin{proof}
	By definition, together with the arguments in Proposition~\ref{pro2}, we only need to check items (SI5) and (SI6) in Definition~\ref{SIM}. Let $\vec{\rho}$ and $\vec{\sigma}$ be two sets of qubit states with $\vec{\rho}\cap\vec{\sigma}=\emptyset$. By definition,
	\beax
	\vmI''_R(\vec{\rho}\cup\vec{\sigma})
	&=&
	\min_{\bfp\in S^2}
	\left(
	\sum_{\rho_j\in \vec{\rho}} |\langle \bfr_{j}|\bfp\rangle|
	+
	\sum_{\rho_j\in \vec{\sigma}} |\langle\bfr_{j} |\bfp\rangle|
	\right)\\
	&\geqslant&
	\min_{\bfp\in S^2}\sum_{\rho_j\in \vec{\rho}} |\langle\bfr_{j}|\bfp\rangle|
	+
	\min_{\bfp\in S^2}\sum_{\rho_j\in \vec{\sigma}} |\langle\bfr_{j}|\bfp\rangle|\\
	&=&\vmI''_R(\vec{\rho})+\vmI''_R(\vec{\sigma}).
	\eeax
	That is, $\vmI''_R$ is complete.
	
	Taking
	$\vec\rho=\{\rho_1,\rho_2,\rho_3\}$ and $\vec\sigma=\{\rho_1,\rho_2,\rho_4\}$ with the Bloch vectors are $\bfr_1=(1,0,0)$, $\bfr_2=(0,1,0)$, $\bfr_3=(0,0,1)$, and $\bfr_4=(0,0,-1)$, respectively,
	we have, according to Eq.~\eqref{eq1},
	\begin{align*}
		&\vmI''_R(\vec\rho)=\min_{\bfp\in S^2}\bigl(|p_x|+|p_y|+|p_z|\bigr)=1, \\
		&\vmI''_R(\vec\sigma)=1,~~
		\vmI''_R(\vec\rho\cap\vec\sigma)=0.
	\end{align*}
	If we take $\bfp=(1, 0, 0)$, then, by $|p_x|+|p_y|+2|p_z|\geqslant |p_x|+|p_y|+|p_z|\geqslant 1$, we obtain
	\bex
	\vmI''_R(\vec\rho\cup\vec\sigma)
	=\min_{p\in S^2}\bigl(|p_x|+|p_y|+2|p_z|\bigr)=1,
	\eex
	which leads to
	\bex
	\vmI''_R(\vec\rho)+\vmI''_R(\vec\sigma)
	>
	\vmI''_R(\vec\rho\cup\vec\sigma)+\vmI''_R(\vec\rho\cap\vec\sigma).
	\eex
	So item (SI6) fails.
\end{proof}


\section{Contractive metric induced qubit set imaginarity}


For any metric $\mC$ that is contractive under CPTP maps, Miyazaki and Matsumoto~\cite{Miyazaki2022q} defined mean-distance
\be
\vmI_C(\vec{\rho}) = \min_{U} \frac{1}{n} \sum_{j=1}^{n} \mI_C(U\rho_jU^\dag)
\ee
and max-distance
\be
\vmI'_C(\vec{\rho}) = \min_{U} \max_{j} \mI_C(U\rho_jU^\dag)
\ee
with respect to any given reference basis. We can also define
\be
\vmI''_C(\vec{\rho}) = \min_{U}  \sum_{j=1}^{n} \mI_C(U\rho_jU^\dag)
\ee
which is deduced by removing the factor $\frac{1}{n}$ from $\vmI_C$.

\begin{pro}
$\vmI_C$, $\vmI'_C$ and $\vmI''_C$ are QbSIMs.
\end{pro}

\begin{proof}
	By definition, we only need to check item (SI3) in Definition~\ref{SIM}.
	If $\mE$ is planarized, then the image of the Bloch ball under $\mE$ is contained in an affine plane passing through the origin. Hence, for any set $\vec\rho$, all states $\mE(\rho_j)$ lie in a common plane through the origin in the Bloch sphere. So there exists a unitary $V$ such that this plane is rotated to the fixed real plane associated with the reference basis, namely $V\mE(\rho_j)V^\dagger\in\mR$ for all $j$. Hence $
	\mI_C(V\mE(\rho_j)V^\dagger)=0$ and thus
	\beax
	\vmI_C(\mE(\vec\rho)) \le \vmI_C(\vec\rho).
	\eeax
	
	If $\mE$ is unital, then the set $U^\dagger \mR U=\{U^\dagger\sigma U:\sigma\in\mR\}$ corresponds to a plane through the origin in the Bloch ball. Since $\mE$ is unital, the image of this plane under $\mE$ is again contained in a linear subspace of dimension at most two passing through the origin. Therefore, there exists a unitary $V$ such that
	\bea
	V \mE(U^\dagger \mR U) V^\dagger \subseteq \mR.
	\eea
	Let $\mE'(\rho)=V \mE(U^\dagger \rho U) V^\dagger$. We have
	$\mE'(U\rho_jU^\dagger)=V\mE(\rho_j)V^\dagger$ and $\mE'(\mR)\subseteq\mR$.
	By the monotonicity of $\mI_C$, we get
	\beax
	\mI_C(\mE'(U\rho_jU^\dagger))\le \mI_C(U\rho_jU^\dagger),
	\eeax
	and therefore
	\beax
	\mI_C(V\mE(\rho_j)V^\dagger) \le \mI_C(U\rho_jU^\dagger),~\forall\, j.
	\eeax
	Then we obtain
	\beax
	\vmI_C(\mE(\vec\rho)) \le \vmI_C(\vec\rho).
	\eeax
	Thus $\vmI_C$ satisfies (SI3). This completes the proof.
\end{proof}

Going further, we can easily check that $\vmI'_C$ is unified and $\vmI''_C$ is complete.


\section{Conclusion and discussion}


We have ascertained the free operation of the set imaginarity for the qubit system, from which we established the axiomatic definition of the set-imaginarity measure. Unlike individual-state resource theories, the free operations for set-states exhibit fundamentally distinct behavior. For the higher-dimensional case, the exact form of such a free operation is hard to fix in general. For the qubit system, we developed several measures of set imaginarity, especially the one based on the Bargmann invariants. The robustness-based approach effectively quantifies both individual-state imaginarity and set imaginarity.

The trace distance induced contractive metric had been generalized to mean-distance, sum-distance and max-distance for set imaginarity, which coincide with  $\vmI_{R}$, $\vmI'_{R}$, and $\vmI''_{R}$, respectively. We proved that any measure of imaginarity $\mI_C$ that induced by contractive metric can be converted into a measure of set imaginarity. These induced SIMs possibly have the same properties as $\vmI_{R}$, $\vmI'_{R}$, and $\vmI''_{R}$, respectively. But the detection and quantifying of set imaginarity in the higher-dimension system remains challenging (see Appendix~\ref{c} for more details). We defer this important direction to future studies.

The resource theory of the set-state emerges as a framework beyond conventional nonconvex resource theories of the non-Gaussianity, non-Markovianity, and quantum discord~\cite{Chitambar2019rmp}. Our findings pave the way for future research in these new nonconvex quantum resource theories and advance the understanding of collective imaginarity in quantum world.

\begin{acknowledgements}
Y.G. is supported by the National Natural Science Foundation of China under Grants No.~12471434 and No.~11971277, the Program for Young Talents of Science and Technology in Universities of Inner Mongolia Autonomous Region under Grant No. NJYT25010, the High-Level Talent Research Start-up Fund of Inner Mongolia University under Grant No. 10000-A260015/501, and the Inner Mongolia Autonomous Region Science and Technology Plan Projects under Grant No. 2025KYPT0098. S.D. is supported by the National Natural Science Foundation of China under Grant No.~12271452.
\end{acknowledgements}

\appendix


\section{Coarsening relation of multipartite partition}\label{a}


Hereafter, we denote by $A_1A_2\cdots A_n$ an $n$-partite quantum system with the state space $\mH^{A_1A_2\cdots A_n}=\mathcal{H}^{A_1}\otimes \mathcal{H}^{A_2}\otimes\cdots\otimes\mathcal{H}^{A_n}$, where $\mH^{A_i}$'s are Hilbert spaces with finite dimension, and by $\mS^{X}$ we denote the set of all density operators (or called states) acting on $\mH^{X}$. The superscript or subscript $X$ always denotes the corresponding system. For example, the state in $\mS^{X}$ is denoted by $\rho^X$ (or $\rho_X$ sometimes), and is also denoted by $\rho$ for simplicity whenever the associated system $X$ is clear from the context.
$X_1|X_2| \cdots |X_{k}$ denotes the $k$ partition of $A_1A_2\cdots A_n$ (or subsystem of $A_1A_2\cdots A_n$ sometimes), $k\leqslant  n$. For instance, partition $AB|C|DE$ is a three partition of the five-particle system $ABCDE$ with $X_1=AB$, $X_2=C$ and $X_3=DE$. The case of $k=n$ is just the original $n$-particle system without any other partition, namely, $A_1A_2\cdots A_n$ means $A_1|A_2|\cdots |A_n$. So, in general $k< n$ unless otherwise specified. We denote the set of all the $k$ partitions of $A_1A_2\cdots A_n$ by $\Gamma_k$ as in Ref.~\cite{Guo2024pra}, $2\leqslant  k<n$, i.e., $\Gamma_k=\{\gamma_i\}$, where $\gamma_i=X_{1(i)}|X_{2(i)}|\cdots|X_{k(i)}$.
Let $\gamma$ and $\gamma'$ be two partitions of $A_1A_2\cdots A_n$ or subsystem of $A_1A_2\cdots A_n$, $k\leqslant  n$, $l\leqslant  n$. We denote by~\cite{Guo2022entropy,Guo2024pra,Guo2025arxiv}
\bea
\gamma\succ^a \gamma', ~
\gamma\succ^b \gamma',~
\gamma\succ^c \gamma'
\eea
if $\gamma'$ can be obtained from $\gamma$
by
\begin{itemize}
	\item[(a)] discarding some subsystem(s) of $\gamma$,
	\item[(b)] combining some subsystems of $\gamma$,
	\item[(c)] discarding some subsystem(s) of some subsystem(s) $X_t$ provided that $\gamma=X_1|X_2| \cdots| X_{k}$, $X_{t}=A_{t(1)}A_{t(2)}\cdots A_{t(f(t))}$ with $f(t)\geq2$, $1\leqslant  t\leqslant  k$,
\end{itemize}
respectively. For example,
$A|B|C|D\succ^a A|B|D\succ^a B|D$,
$A|B|C|D\succ^b AC|B|D\succ^b AC|BD$,
$A|BC\succ^c A|B$.
We call $\gamma'$ is coarser than $\gamma$ if
$\gamma'$ can be obtained from $\gamma$
by one or some of the ways in item (a)--item (c), and we denote it by
$\gamma\succ \gamma'$
uniformly.

Furthermore, if $\gamma\succ \gamma'$, then we denote by
\bea \label{Xi}
\Xi(\gamma- \gamma')
\eea
the set of all the partitions that are coarser than $\gamma$ but (i) neither coarser than $\gamma'$ nor the one from which one can derive $\gamma'$ by the coarsening means, and (ii) if it includes some or all subsystems of $\gamma'=Y_1|Y_2| \cdots |Y_{l}$, then all the subsystems $Y_j$'s included are regarded as one subsystem and (iii) if $\gamma'=Y_1|Y_2| \cdots |Y_{l}$ and $\gamma=X_{l}|X_2|\cdots| X_{k}$ with $Y_1|Y_2| \cdots |Y_{l}=X_1|X_2| \cdots|X_{l-1}|X_{l}\cdots X_{k}$, then $\Xi(\gamma- \gamma')$ contains only $X_{l}|\cdots| X_{k}$ and the one coarser than it. For example, $\Xi(A|B|C|D-A|BCD)=$$\{B|C|D$, $B|CD$, $BC|D$, $C|BD$, $B|C$, $C|D$, $B|D\}$.

\section{Complete global multipartite entanglement measure}\label{b}

Recall that, a function $E^{(n)}: \mS^{A_1A_2\cdots A_n}\to\mbb{R}^{+}$ is called a $n$-partite entanglement measure~\cite{Horodecki2009rmp,Hongyan2012pra,Guo2020pra} if it satisfies: $(E1)$ $E^{(n)}(\rho)=0$ if $\rho$ is fully separable; $(E2)$ $E^{(n)}$ cannot increase under $n$-partite LOCC. In addition, $E^{(n)}$ is said to be a $n$-partite entanglement monotone if it is convex and does not increase on average under $n$-partite stochastic LOCC.

Going further, an MEM $E^{(n)}$ is called a {unified}
global multipartite entanglement measure (GlMEM) if it satisfies the unification condition~\cite{Guo2020pra,Guo2024rip,Guo2025arxiv}:

(i) (Superadditivity)
\begin{equation}\label{supadditivity}
\begin{aligned}
&E^{(n)}(A_1A_2\cdots A_k\ot{A_{k+1}\cdots A_n})\\
\geqslant& E^{(k)}({A_1A_2\cdots A_k})+E^{(n-k)}({A_{k+1}\cdots A_n}),
\end{aligned}
\end{equation}
holds for all $\rho^{A_1A_2\cdots A_k}\ot\rho^{A_{k+1}\cdots A_n}\in\mS^{A_1A_2\cdots A_n}$, hereafter $E^{(n)}(X)$ refers to $E^{(n)}(\rho^X)$ and $E^{(1)}=0$ [note here that, in Refs.~\cite{Guo2020pra,Guo2024rip}, the condition in Eq.~\eqref{supadditivity} is restricted to additivity, i.e., $E^{(n)}(A_1A_2\cdots A_k\ot{A_{k+1}\cdots A_n})
=E^{(k)}({A_1A_2\cdots A_k})+E^{(n-k)}({A_{k+1}\cdots A_n})$.
We weakened it to superadditivity which includes the additivity as a special case];

(ii) (Symmetry) $E^{(n)}({A_1A_2\cdots A_n})=E^{(n)}({A_{\pi(1)}A_{\pi(2)}\cdots A_{\pi(n)}})$, for any $\rho\in\mS^{A_1A_2\cdots A_n}$ and any permutation $\pi$;

(iii) (Coarsening monotone)
\bea\label{coarsen}
E^{(k)}(\gamma)\geqslant E^{(l)}(\gamma')
\eea
holds for all $\rho\in\mS^{A_1A_2\cdots A_n}$ whenever $\gamma\succ^a \gamma'$, where $\gamma$ and $\gamma'$ are two partitions of $A_1A_2\cdots A_n$ or subsystem of $A_1A_2\cdots A_n$.

$E^{(n)}$ is called a {complete} GlMEM if it satisfies both the unification condition above and the hierarchy condition~\cite{Guo2020pra}:

(iv) (Tight coarsening monotone) Equation~\eqref{coarsen} holds for all $\rho\in\mS^{A_1A_2\cdots A_n}$ whenever $\gamma\succ^b \gamma'$.

By definition, a unified MEM $E^{(n)}$ in fact refers to a family of measures $\{E^{(k)}: 2\leqslant k\leqslant n\}$. For example, if $E^{(n)}$ is unified, and we consider the tripartite system, then for any state, $E^{(3)}(ABC)\geqslant E^{(2)}(AB)$, and if it is complete, then $E^{(3)}(ABC)\geqslant E^{(2)}(A|BC)$. Not all MEMs are unified and some unified GlMEMs are not complete~\cite{Guo2024rip}.


\section{Set imaginarity of qudit states}\label{c}


Let $\vec{\rho}\subset\mS(3)$ with $|\vec{\rho}|=4$ and $\rho_i=|\psi_i\ra\la\psi_i|$, where
\bax
|\psi_1\ra&=|0\ra,~|\psi_2\ra=\frac{1}{\sqrt2}(|0\ra+|1\ra),\\
|\psi_3\ra&=\frac{1}{\sqrt2}(|1\ra+|2\ra),~|\psi_4\ra=\frac{1}{\sqrt3}(|0\ra-|1\ra+\rmi|2\ra).
\eax
Then we can easily check that all the third-order Bargmann invariants are zero but $\tr(\rho_1\rho_2\rho_3\rho_4)=-1+\rmi$. Namely, the third-order Bargmann invariants cannot detect the imaginarity in $\vec{\rho}$. But for some special set $\vec{\rho}$, we can detect the set imaginarity by the third-order Bargmann invariants. Let $\{\ket{\psi_i}\}_{i=1}^n\subset\mH$ ($\dim\mH=d\geqslant2$) be a family of pure states whose support graph is complete, i.e.,
\[
\la\psi_i|\psi_j\ra\neq 0
\quad
\text{for all }i\neq j.
\]
Then the $k$th order Bargmann invariant of $\{\rho_i=|\psi_i\ra\la\psi_i|\}_{i=1}^k$ satisfies
\bea
&&\Delta_k(\rho_1\rho_2\cdots\rho_k)\nonumber \\
&=&
\frac{
	\Delta_3(\rho_1,\rho_2,\rho_3)\,
	\Delta_3(\rho_1,\rho_3,\rho_4)\,
	\cdots\,
	\Delta_3(\rho_1,\rho_{k\text{-}1},\rho_k)
}{
	\prod_{m=3}^{k-1}
	|\la\psi_{1}|\psi_{m}\ra|^2
}.~~~~~~~
\eea
That is, in such a case, if only one or three third-order Bargmann invariants are not real, then the set has set imaginarity.

We can represent any bipartite state $\rho\in\mS(d)$ in the Bloch form (also known as Fano form~\cite{Fano1983rmp}) by
\beax
\rho &=& \frac{1}{d}\left(  I
+ \sum_{i=1}^{d^2-1} r_{i} \frac{\lambda_{i}}{\sqrt{2}}\right),
\eeax
where
\beax
r_{i} = \tr\left(\rho\frac{\lambda_{i}}{\sqrt{2}}\right),
\eeax
$\{\lambda_i\}_{i=1}^{d^2-1}$ is the
traceless Hermitian generators of $SU(d)$. 
\if false Note here that $\{\lambda_i\}=\{w_l,u_{jk},v_{jk}\}$
can be constructed from any orthonormal basis in $\mH$ \cite{suppl-Hioe1981prl},
\begin{align*}
	&w_l=\sqrt{\frac{2}{(l+1)(l+2)}}\left(\sum_{i=0}^l|i\rangle\langle
	i|-(l+1)|l+1\rangle\langle l+1|\right),\nonumber\\
	&u_{jk}=|j\rangle\langle k|+|k\rangle\langle j|,\nonumber\\
	&v_{jk}=-\rmi|j\rangle\langle k|-|k\rangle\langle j|,\nonumber
\end{align*}
where $0\leqslant l\leqslant d-2$ and $0\leqslant j<k\leqslant d-1$.
\fi
In such a sense, any $\rho\in\mS(d)$ corresponds to a $d^2-1$-dimensional real vector $\bfr$ which is called the generalized Bloch vector of $\rho$. For any given $\vec{\rho}\in\mS(d)$ with $|\vec{\rho}|=n$, Li \emph{et al}. define the associated Gram matrix by~\cite{LiMS2026pra}
\beax
G_{\mathbf{r}(\vec{\rho})} := \left[ \langle \mathbf{r}_k|\mathbf{r}_l \rangle\right],
\eeax
where $\bfr_k$ is the generalized Bloch vector of $\rho_k$, and it was shown that,
if $\vec{\rho}$ is set imaginarity free, then $\operatorname{rank}(G_{\mathbf{r}(\rho)}) \leqslant \frac{d(d+1)}{2} - 1$.

From the argument above, we give the following conjecture.

\begin{conjecture}
	Let $\vec{\rho}=\{\rho_j\}_{j=1}^n\subset \mS(d)$ be a set of qudit states. Then $\vec{\rho}$ is set imaginarity free if and only if there exist $i_1$, $i_2$, $\dots$, $i_k$, $3\leqslant k\leqslant d(d+1)/2$, such that
	\bex
	\rIm\tr(\rho_{i_1}\rho_{i_2}\cdots\rho_{i_k})\neq0.
	\eex
\end{conjecture}



\end{document}